\documentclass[letterpaper, 10 pt, conference]{ieeeconf} % Comment this line out if you need a4paper

\IEEEoverridecommandlockouts

\usepackage[active]{srcltx}

\usepackage{pdfpages}
\usepackage{verbatim}
\usepackage{epsfig}
\usepackage{amsmath}
\usepackage{amssymb}
\usepackage{psfrag}
\usepackage[T1]{fontenc}
\usepackage{graphicx}
\usepackage[usenames,dvipsnames]{pstricks}
\usepackage{pst-grad} % For gradients
\usepackage{pst-plot} % For axes
\usepackage{algorithm}
\usepackage{algpseudocode}
\usepackage{mathtools}
\usepackage{setspace}
\usepackage{url}
\let\labelindent\relax
\usepackage{enumitem}
\usepackage{amsmath}
\newcommand{\subparagraph}{}
\usepackage{titlesec}

\graphicspath{{./figure/}}

\newtheorem{theorem}{Theorem}[section]

\newtheorem{remark}{Remark}

\title{\LARGE \bf
Data-driven identification of a thermal network in multi-zone building}

\author{Harish Doddi$^1$, Saurav Talukdar$^{1}$, Deepjyoti Deka$^{2}$  and Murti Salapaka$^{3}$ \\
\thanks{$^{1}$Harish Doddi and Saurav Talukdar is with Department of Mechanical Engineering, University of Minnesota, Minneapolis, USA,
{\tt\small doddi003@umn.edu, sauravtalukdar@umn.edu}}%
\thanks{$^{2}$Deepjyoti Deka is with Theory Division, Los Alamos National Laboratory, Los Alamos, USA,
{\tt\small deepjyoti@lanl.gov}}%
\thanks{$^{3}$Muti V. Salapaka are with Department of Electrical and Computer Engineering, University of Minnesota, Minneapolis, USA,
{\tt\small murtis@umn.edu}}
}

\begin{document}
\maketitle
\thispagestyle{empty}
\pagestyle{empty}

%%%%%%%%%%%%%%%%%%%%%%%%%%%%%%%%%%%%%%%%%%%%%%%%%%%%%%%%%%%%%%%%%%%%%%%%%%%%%%%%
\begin{abstract}
System identification of smart buildings is necessary for their optimal control and application in demand response. The thermal response of a building around an operating point can be modeled using a network  of interconnected resistors with capacitors at each node/zone called RC network. The development of the RC network involves two phases: obtaining the network topology, and estimating thermal resistances and capacitance's. In this article, we present a provable method to reconstruct the interaction topology of thermal zones of a building solely from temperature measurements. We demonstrate that our learning algorithm accurately reconstructs the interaction topology for a $5$ zone office building in EnergyPlus with real-world conditions. We show that our learning algorithm is able to recover the network structure in scenarios where prior research prove insufficient.
\end{abstract}

%%%%%%%%%%%%%%%%%%%%%%%%%%%%%%%%%%%%%%%%%%%%%%%%%%%%%%%%%%%%%%%%%%%%%%%%%%%%%%%%
\section{Introduction}
As complex cyber-physical systems, commercial buildings are one of the largest consumers of electricity in the United States. According to \cite{energy2010energy}, there are more than 4.7 million commercial buildings and 114 million households that together consume nearly 40$\%$ of the total power produced. Moreover they contribute to 40$\%$ of $CO_{2}$ emissions according to U.S Green Building Council \cite{energy2011buildings}. Optimal control of smart buildings is thus an important focus of smart grid research and deployment. In a commercial building, Heating Ventilation and Air Conditioning (HVAC) systems is the largest consumer of power. Hence there is increased research interest in designing zero-energy and smart buildings. Such goals are targeted with the additional objective of ensuring occupant's comfort. Strategies leverage flexible operations and integration of renewable resources using techniques such as building pre-cooling and coordinated dynamic control \cite{nagarathinam2017energy} based on load prediction. Advanced predictive HVAC control schemes such as Model Predictive Control (MPC) are reported to achieve 12-15\% savings in energy consumption without compromising on occupant thermal comfort \cite{nagarathinam2017energy}. The efficiency of MPC and other real-time control policies depend on an accurate and tractable thermal model of the concerned building. In particular, changing occupancy and weather conditions over time (medium to long term) may affect the necessary thermal models and hence modify the optimal building controls. Developing an accurate thermal model of smart buildings is thus a necessary step for their control. It needs to be mentioned that in typical commercial buildings, the Building Management/Automation System monitors all the operations of the building such as door access, thermostat set-points, HVAC operating points, lighting schedules and interfacing to sensors and actuators. 
Modeling of buildings can be categorized into several sub-types. Effective simulation tools for buildings include EnergyPlus \cite{crawley2000energyplus} and Trnsys \cite{trnsys2000transient}. These simulation engines are called \emph{white box models} as they consider extensive details of the building physics and associated HVAC systems in the thermal models. EnergyPlus, in particular, is a widely accepted standard in industry for energy simulation, equipment sizing and design of controllers. It can also be integrated with other computational engines to test sophisticated controllers using Building Controls Virtual Test Bed \cite{wetter2008building}. Although EnergyPlus simulates the thermal response of the building accurately, it involves several (hundreds) physical parameters and performing time-consuming Computational Fluid Dynamics (CFD) simulations to determine heat transfer coefficients is nearly impossible for real-time control. This complexity of such high-fidelity engines makes real-time model update and inference computationally prohibitive - more so in multi-zone buildings. 

In contrast to white box models, there are two other types of modelling approaches, namely \emph{black box models} and \emph{grey box models}. Black box models  use techniques such as linear regression, neural networks \cite{karatasou2006modeling} to fit the available input/output data. Indeed their performance depends greatly on the quality of the input and output data. Grey box models \cite{reynders2014quality} uses a hybrid mix of white box and black box modeling, and are used to obtain fast estimates of coarse-grained system conditions. Grey box models consider some prior information of system dynamics and augment them with measurements for system identification. One of the well-known grey-box models for thermal systems/buildings is the Resistor-Capacitor (RC) model, where a thermal circuit is used to represent the heat transfer dynamics in the building.
Prior work on grey box modeling of building using RC networks assumed the structure of the RC network (based on experience/ physical insight of the expert) \cite{goyal2010modeling} and utilize time series measurements from the building zones to estimate the resistance and capacitance parameters of the RC network. However, in the case of open-plan office building, the structure of the RC network is transient and unknown. \cite{6161387} uses a two step approach to obtain the RC network: in the first step the authors identify the network structure using approaches developed for structure learning of Gaussian graphical models followed by determination of the resistance and capacitance parameters from the input/output data. The drawback of the first step is that, it performs poorly when samples are correlated, which is the case in real scenarios.

RC networks are an example of undirected networked linear dynamical systems. Learning the structure of network of linear dynamical systems from nodal time series measurements is an active area of research \cite{materassi2012problem,dankers2016identification, gonccalves2008necessary, pereira2010learning}. These works primarily focus on directed networks of linear dynamical systems or assume uncorrelated inputs. Structure learning in undirected linear systems has been explored recently for radial networks \cite{talukdar2017radial} and loopy networks \cite{talukdar2017learning,talukdar2017consensus} in power systems and multi-agent distributed systems. These articles utilize properties of multivariate Wiener filters to provide consistent structure estimation. We employ a similar approach in this article, and present an algorithm for exact topology reconstruction of RC networks from nodal temperature measurements. A distinguishing feature of the approach presented is that heat generation or loss term (exogenous disturbance) at each node is allowed to be a colored input, unlike white time series in previous work \cite{pereira2010learning}. We also develop a regularized version of the method to improve the accuracy of the learning algorithm in low sample regime. We are restricting to sparsity based regularization. The work presented in this article applies to non-positive systems too. Note that, no additional knowledge about the building structure or information (samples or sufficient statistics) of system inputs are assumed in our learning framework.

To demonstrate the efficacy of our algorithm, we consider a  five zone building model in EnergyPlus with an equivalent RC-network (derived from physics) given in \cite{6315699}. We show that the topology obtained by using our algorithm on zonal temperature data from its rooms with correlated input disturbances matches the true RC network topology. We show that our method outperforms learning algorithms for uncorrelated inputs \cite{pereira2010learning} and static graphical models used in prior work \cite{6161387}. We also demonstrate performance improvements in our learning algorithm with low samples due to regularization.

The remaining paper is organized as follows. Section II describes the Resistor-Capacitor Model. In Section III, the description of our EnergyPlus model is given. In Section IV, framework of the topology learning algorithm is presented. Results of our algortihm are discussed in Section V followed by conclusion in Section VI. 

\section{Resistor-Capacitor (RC) Model}
Consider a building with $m$ thermal zones exchanging energy with each other through conduction, convection and radiation. The thermal dynamics around an operating point can be modeled using a RC network, which is defined as a set of zones and edges, where edges are resistors and nodes are modelled as capacitors (see Fig.~\ref{fig:energyplus_5node}(b)). The topology of the RC network is defined as an undirected graph $\mathcal{G}=(\mathcal{V},\mathcal{E})$ where node set $\mathcal{V} = \{1,2,...,m\}$ represents zones and $\mathcal{E}$ denotes the edge set with $(j,i) \in \mathcal{E}$ for $i,j \in \mathcal{V}$ if $R_{ij} \neq 0$ (see Fig.~\ref{fig:energyplus_5node}(c)). In the undirected graph $\mathcal{G}$, \emph{neighbors} of node $j$ are defined as the elements of its neighborhood set, ${N}_j:=\{i\in \mathcal{V}:(i,j)\in \mathcal{E}\}$. The \emph{two hop neighbors} of node $j$ are defined as the elements of the set, ${N}_{j,2}:=\{i \in \mathcal{V}: (j,k),(i,k) \in \mathcal{E}, \text{for some} \ k \in {\cal V}\}$. 

\begin{figure}[tb]
	\centering
	\begin{tabular}{cc}
		\includegraphics[width=0.5\columnwidth, height = 3 cm]{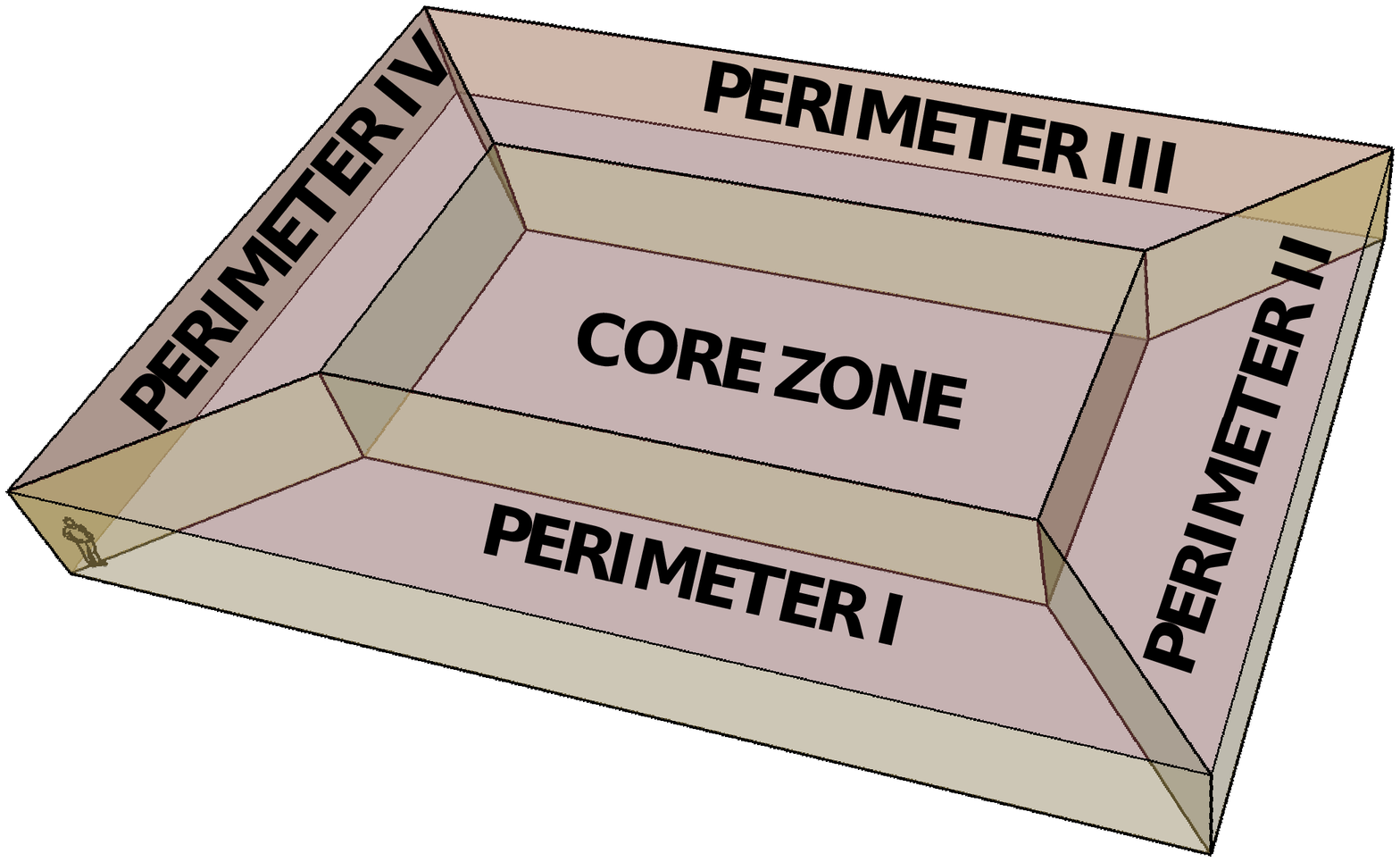}&
		\includegraphics[width=0.3\columnwidth]{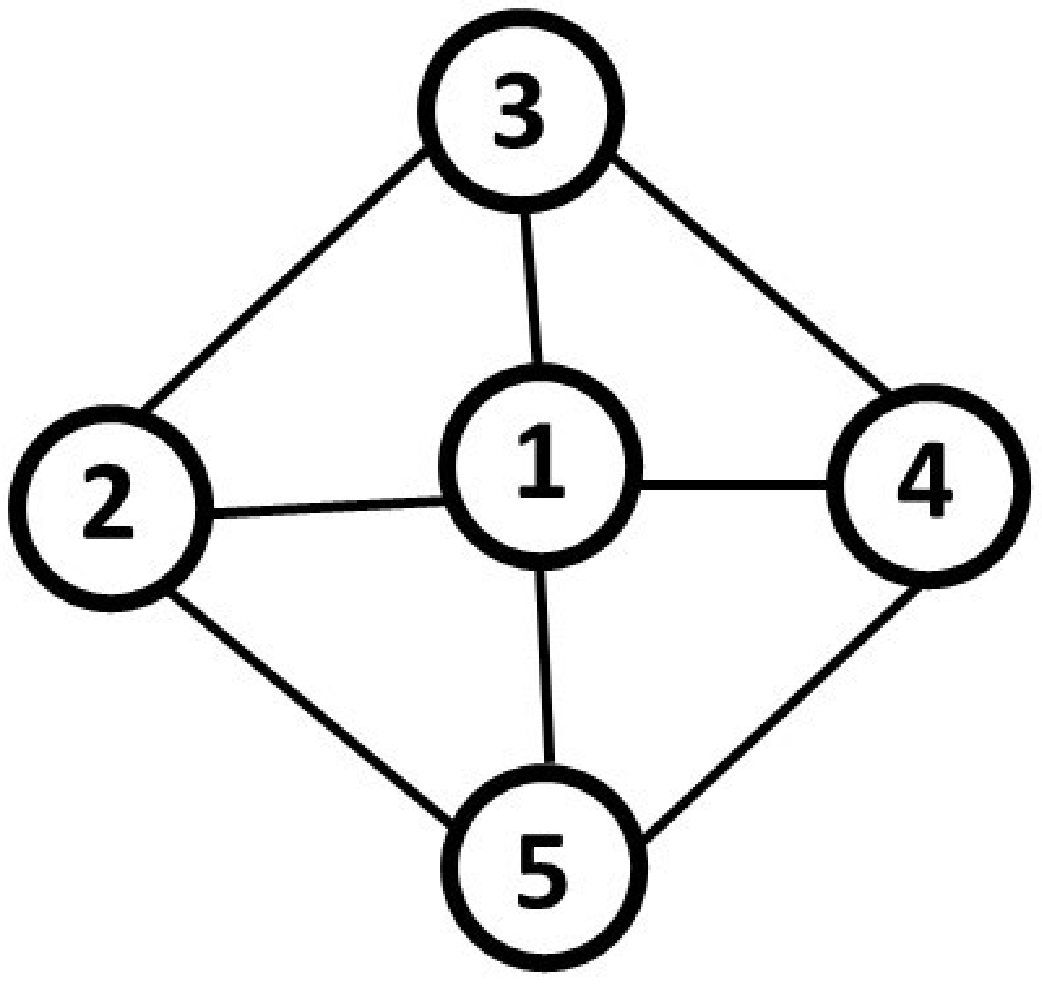}\\
		(a) & (c)
	\end{tabular}
		\begin{tabular}{cc}
		\includegraphics[width=0.4\columnwidth, height = 2.5 cm]{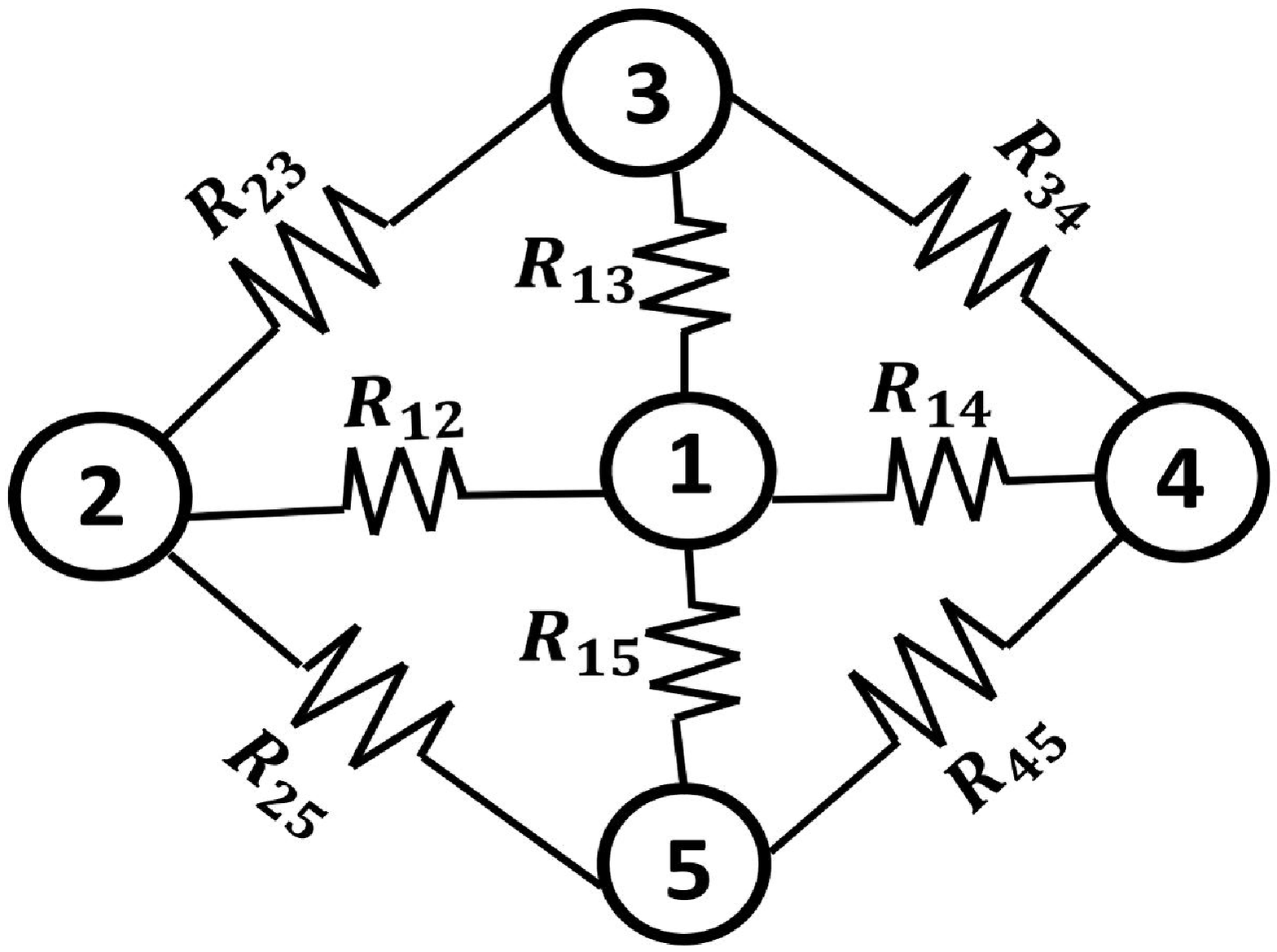}&
		\includegraphics[width=0.35\columnwidth]{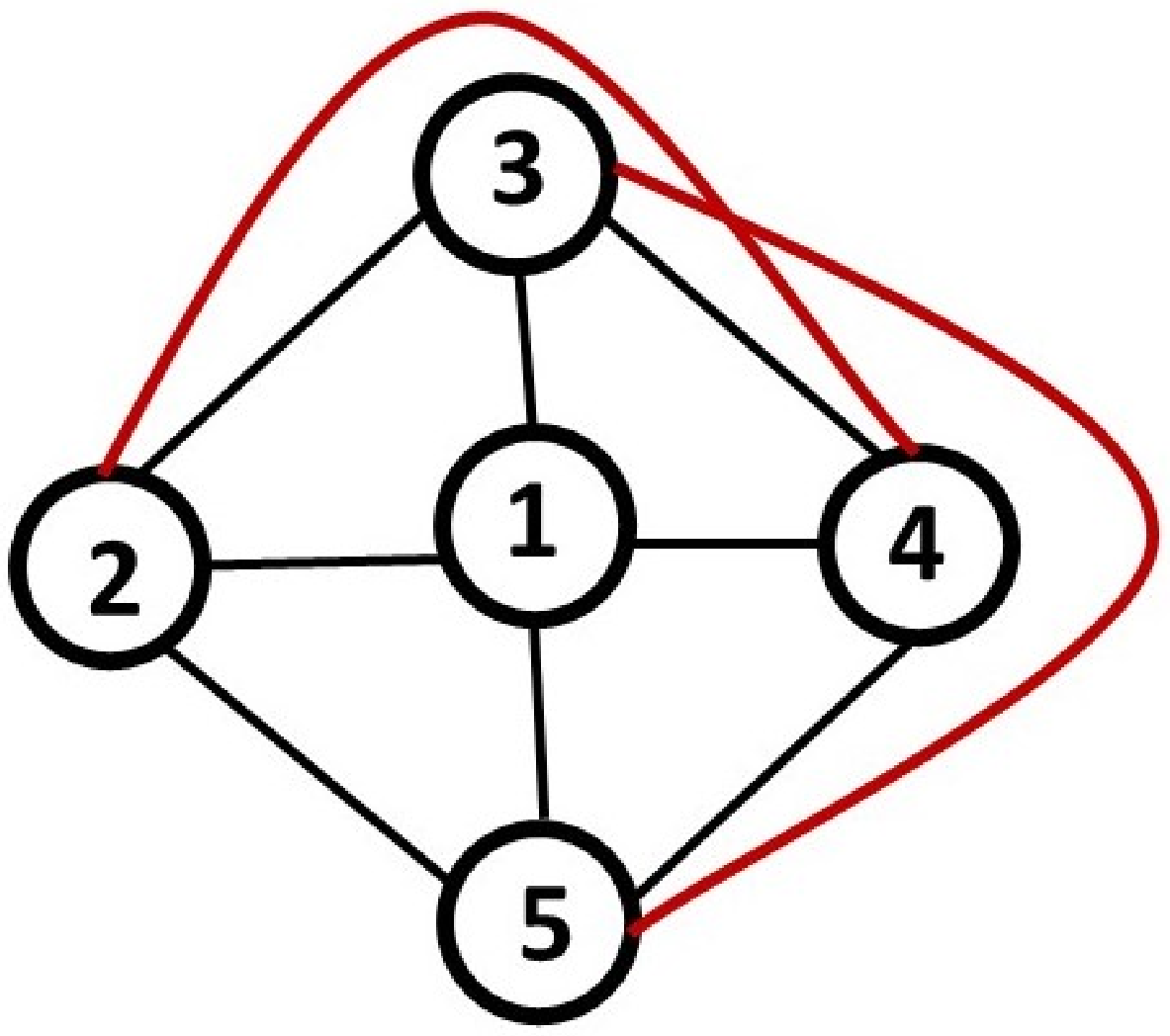}\\
		(b) & (d) 
	\end{tabular}
% 	\caption{An oriented graph (a) and its kin topology 	\end{tabular}
	\caption{(a) EnergyPlus model consisting of $5$ zones, (b) Thermal resistor (RC) network of the building with Capacitor at each node. Node 1 corresponds to core zone and rest of the nodes refer to perimeter zones (c) True topology of the RC network (d) Wiener reconstructed graph}\label{fig:energyplus_5node}
\end{figure}

For thermal modeling, each edge $(j,i)$ is associated with a thermal resistance ${R}_{j,i} = {R}_{i,j}>0$ and each node $i \in \mathcal{V}$ is associated with a capacitor with capacitance $\mathcal{C}_i$ . Let $T_j$ denote the deviation of temperature from the operating/equilibrium point of zone $j$. Then $T_{j}$ satisfies heat balance equation:

\begin{align} \label{eqn:zone}
C_j\frac{dT_j}{dt} = \sum_{i \in N_j}\frac{T_i-T_j}{R_{ji}} + q_j,
\end{align}
where, ${q}_j$ is the total internal heat generated in zone ${j}$. Writing the evolution of temperature deviations at all nodes in matrix form results in the following:
\begin{align}
    [C]\frac{dT}{dt} = [R]T + q \label{eqn:LTI_thermal}
\end{align}
where, $T$ and $q$ are the vectors associated with zonal temperatures and internal heat; $[C]$ is the diagonal matrix of nodal capacitance and $[R]$ denotes the inter-zone resistances ($R_{ij}$). Note that (\ref{eqn:LTI_thermal}) has the form of a networked linear-time invariant (LTI) system with matrix $[R]$  encoding the network information. Learning the topology of the RC-network entails identifying the location of non-zero off-diagonal terms in matrix $R$. Next we discuss the mathematical representation of zonal temperature evolution in EnergyPlus software that we use in subsequent discussion.

\section{EnergyPlus Model}
EnergyPlus is a building energy simulation program developed by U.S. Department of Energy. It is an open-source software available at \url{https://energyplus.net/downloads}. EnergyPlus accepts a text file for providing details on construction of the building, its location, ambient conditions, occupancy, electrical equipment, lighting loads on the building and HVAC. It is a sophisticated forward simulation tool which is useful in thermal analysis of a building.

EnergyPlus assumes each thermal zone as a single node and solves the heat balance equations to arrive at the thermal zone temperatures developed in the building and its power consumption. The heat balance equation for a zone \cite{doe2010energyplus} is given by:
\begin{align}
&C_{z}\frac{dT_z}{dt} = \sum\limits_{i=1}^{N_{sl}}\Dot{Q_{i}} +\sum\limits_{i=1}^{N_{surf}}h_{i}A_{i}(T_{si}-T_{z}) \\
&~+\sum\limits_{i=1}^{N_{zones}}\Dot{m_{i}}C_{p}(T_{zi}-T_{z})+\Dot{m}_{inf}C_{p}(T_{\infty}-T_{z})+\Dot{Q}_{sys}\nonumber\end{align}
where $C_z\frac{dT_z}{dt}$ is the energy stored in zone air, $\sum \Dot{Q_{i}}$ is internal convective load,$\sum h_{i}A_{i}(T_{si}-T_{z})$ is surface convective heat transfer, $\sum\Dot{m_{i}}C_{p}(T_{zi}-T_{z})$ is interzone heat transfer, and
$\Dot{m_{inf}}C_{p}(T_{\infty}-T_{z})$ is heat transfer by air infiltration, $\Dot{Q_{sys}}$ is air systems output.

In this article, we consider a representative building with a building envelope created with Google SketchUp Make 2017 \cite{ellis2008energy}. The RC network of the EnergyPlus model is obtained from {\cite{6315699}} and is shown in Fig. \ref{fig:energyplus_5node}(b). In EnergyPlus, the exogenous inputs (unmeasured) to the building are internal heat loads such as heat generation by lights, electrical equipment and people in a building. The output of the model is the zone air temperature, which is the instantaneous temperature averaged over space for each zone. Note that the different operating parameters and weather data may be temporally correlated. In a real setting, an operational building with Building Management System (BMS) may provide the data, replacing EnergyPlus.

In the next section, we leverage the multivariate Wiener filtering approach \cite{materassi2012problem} to develop an exact learning algorithm of RC networks under wide-sense stationary inputs. In subsequent sections, we show that our learning algorithm is able to recover the structure from the temperature data.

\section{Learning Topology from Temperature Time Series}
To develop our learning algorithm, we first discretize it. The discretized form of continuous time zonal thermal dynamics described by (\ref{eqn:zone}) is given by: 
\begin{align}
T_j(k+1) = a_{jj}T_j(k) + \sum_{i \in N_j} a_{ji}T_i(k) + p_j(k)
\end{align}
where, $p_j(k):=q_j(k)\delta t /C_j$ is the total internal load in the zone $j$ per unit thermal capacitance at instant $k$, $a_{jj}= \delta t(1-\sum_{i \in N_j}\frac{1}{R_{ji}C_j})$,and $ a_{ji}= \frac{\delta t}{R_{ji}C_j}$ 
Here, $\delta t$ is the measurement sampling time and is taken as $1$ min in this article. All $a_{ij}$'s are non-negative. The undirected nature of the interaction topology ($R_{ij} = R_{ji}$) implies that, $a_{ji} \neq 0$ if and only if $a_{ij} \neq 0$. In vector form the network dynamics can be written as,
\begin{align}
 T(k+1)=AT(k)+P(k),
\end{align}

where, $T(k)=[T_1(k), T_2(k),\cdots, T_m(k)]^{T}$, $P(k) = [p_1(k),p_2(k),\cdots, p_m(k)]^T$ and the matrix $A\in \mathbb{R}^{m \times m}$ is such that $A(j,i) = a_{ji}$. Each element of $P(k)$ is assumed to be zero mean wide sense stationary (WSS) process uncorrelated with any other component of $P(k)$. Note that $P(k)$ can have non-zero temporal correlation. Using the $z$ transform of the correlation function of $P(k)$, the power spectral density matrix $\Phi_P(z)$ is derived. From property of WSS signals and spatial un-correlation, it follows that $\Phi_P(z)$  is a diagonal matrix and $\Phi_P(e^{\hat{j}\omega})$ is real and even function of frequency $\omega$. Furthermore, we assume the noise spectrum at any node $j$, $\Phi_{p_j}(z) > 0$ almost surely (referred to as \textbf{topological detectability} condition). 

The discrete time dynamics of the deviation variable ${T}_j$ of the $j^{th}$ node in the $z$ domain 
can be cast as, 
\begin{align}\label{eqn:zdomnode}
		{\mathcal{T}}_j(z)=\sum_{i\in \mathcal{N}_j }\mathcal{H}_{ji}(z){\mathcal{T}}_i(z) + E_j(z), j = 1,2, \cdots, m.
\end{align}
where, $\mathcal{H}_{ji}(z)=\frac{a_{ji}}{S_j(z)}$, $E_j(z) =\frac{1}{S_j(z)}p_j(z)$ with $S_j(z):=z-a_{jj}$. 
If $i$ is not a neighbor of $j$, then $\mathcal{H}_{ji}(z) = 0, \mathcal{H}_{ij}(z) = 0$. In compact form the dynamics of (\ref{eqn:zdomnode}) is written as,
\begin{align}\label{eqn:netdyn}
    {\mathcal{T}}(z)=H(z){\mathcal{T}}(z) + \mathcal{E}(z).
\end{align}
The transfer function matrix $H(z)$ defines the network dynamics and also specifies the RC network topology. The diagonal entries of $H(z)$, $\mathcal{H}_{jj}(z) = 0$. Based on the assumptions on $P(k)$, $\mathcal{E}(z):=[{E}_1(z),\cdots,E_m(z)]^{T}$ is a collection of the filtered version of z-transform of uncorrelated zero mean WSS processes $\{p_1(k),...,p_m(k)\}$ respectively. The power spectral density matrix, $\Phi_E(z)$ is a diagonal matrix, with $\Phi_E(z)(j,j) = \Phi_{E_j}(z) = \Phi_{p_j}(z)/|S_j(z)|^2$. We assume that $I - H(z)$ is invertible almost surely (referred to as \textbf{well posedness} condition).

Let $T_j(k)$ be the temperature deviation from operating point. Consider the following least square optimization problem on the Hilbert space of $\mathcal{L}_2$ random variables, 
\begin{align}
	&\smashoperator[lr]{\inf_{{\{h_{ji}\}}_{i=1,...,m,i\neq j}}}~~ \mathbb{E}(T_j(k)-
	\sum_{i=1,i\neq j}^{m}\sum_{L=-\infty}^{\infty}h_{ji}^{L}T_i(k-L))^2
	, \label{wiener}
		%&=	\inf_{{\{h_{ji}\}}_{i=1,...,m,i\neq j}} \mathbb{E}(\mathcal{T}_j(z)-\sum_{i=1,i\neq j}^{m}\sum_{L=-\infty}^{\infty}h_{ji}^{L}z^{-L}\mathcal{T}_i(z))^2, \label{wiener2}
\end{align}
	where, $h_{ji}=[h_{ji}^{-\infty},...,h_{ji}^{0},...,h_{ji}^{\infty}]$.
	The solution to the above infinite dimensional optimization problem \cite{materassi2012problem} is the \emph{multivariate non-causal Wiener filter}), $W_j(z)=[W_{j,1}(z),...,W_{j,j-1}(z),W_{j,j+1}(z),...,W_{j,m}(z)]$, where,
	\begin{align}\label{solform}
	    W_{ji}(z) = \sum_{L=-\infty}^{\infty}h_{ji}^{L}z^{-L}. 
	\end{align}
	It is important to note that the multivariate Wiener filter, the solution to optimization problem (\ref{wiener}), can be determined from the temperature time series measurements without knowledge of the topology or the parameters of the dynamics. The next result connects Wiener filters with the underlying topology of RC networks.  

\begin{theorem}\label{thm:sparse_wiener}
	Consider the matrix of Wiener filters, $W(z)$ such that, $W(j,i)(z):=W_{ji}(z)$ which is the solution to (\ref{wiener}) and has the form (\ref{solform}). Then, $W(j,i)(z) \neq 0$ almost surely implies $i,j$ are neighbors or two hop neighbors of each other.
\end{theorem}
The proof follows from applying the main result of \cite{materassi2012problem} to undirected (bi-directed) graph with linear dynamics.
\begin{remark}
 The above result does not guarantee that if $i \in \mathcal{N}_j \cup \mathcal{N}_{j,2}$, then $W_{ji}(z) \neq 0$. However, such cases are pathological (see \cite{materassi2012problem}).
\end{remark}

 Using the non-zero entries in $W(z)$, a graph $\mathcal{G}_M$ (also known as moral graph) can be constructed with the vertex set $V$. The graph $\mathcal{G}_M$, has all the edges in the underlying topology of the RC network with additional spurious edges due to two hop neighbors. Pruning out the spurious two hop neighbor edges will result in recovery of the underlying RC network topology. Next we present results which will enable us to prune out the spurious two hop neighbor links from $\mathcal{G}_M$. 
 
\subsection{Pruning Out Spurious Two-hop Neighbors Links}
\vspace{-.1cm}
Our pruning step is based on the phase response of Wiener filters. We will use the following relationship between $W_{ji}(z)$ and $\Phi_{T}(z)^{-1}(j,i)$ as derived in \cite{materassi2012problem},
\begin{align}\label{eqn:wienerphase}
 W_{ji}(z) = -\Phi_T^{-1}(z)(j,i)\Phi_{E_j}(z).  
\end{align}
It follows from (\ref{eqn:netdyn}) that,
\begin{align}
    \Phi_{T}^{-1}(z) = (I-H(z))^{*}\Phi_E^{-1}(z)(I-H(z)),
\end{align}
where, $I$ is a $m\times m$ identity matrix. Then,
\begin{align}\label{eqn:Phinv}
\Phi_{T}^{-1}(z)(j,i) &= -\mathcal{H}_{ji}(z)\Phi_{E_j}^{-1}(z)-\mathcal{H}_{ij}^*(z)\Phi_{E_i}^{-1}(z) \nonumber \\
&+ \sum_{k \in N_{j}\cap N_{i}} \mathcal{H}_{kj}^*(z)\mathcal{H}_{ki}(z)\Phi_{E_k}^{-1}(z).
\end{align}
\begin{theorem}\label{thm:expression}
Let zones $i$ and $j$ in a RC network be such that, $i \in N_{j,2}$ but $i\notin N_{j}$,i.e., $i,j$ are strict two-hop neighbors. Then, $\angle W_{ji}(e^{\iota \omega}) = \pi$, for all $\omega \in [0, 2 \pi)$.
\end{theorem}

\begin{proof}
Since, $i$ and $j$ are not neighbors in the underlying RC network, then $\mathcal{H}_{ji}(z) = 0$ and $\mathcal{H}_{ij}(z)=0$. Moreover, $i$ and $j$ are two hop neighbors, so $N_j\cap N_i$ is non empty. Thus, using (\ref{eqn:Phinv}) and substituting $z$ with $e^{\iota \omega}$, it follows from (\ref{eqn:wienerphase}), that if $i$ and $j$ are two hop neighbors and are not neighbors, then $\angle W_{ji}(e^{\iota \omega}) = \pi$, for all $\omega \in [0, 2 \pi)$.
\end{proof}

The above theorem can identify spurious edges between nodes which are strict two hop neighbors if it does not hold for nodes that are true neighbors. The next theorem lists conditions under which $\angle (W_{ji}(e^{\hat{j}\omega})) = \pi$ for all $\omega \in [0,2\pi)$ for nodes $i$ and $j$ that are neighbors.

\begin{theorem}\label{thm:notpiresult2}
Given a well-posed and topologically detectable RC network the following holds:
\begin{enumerate}[leftmargin =*]
    \item Suppose nodes $i$ and $j$ are such that $i \in N_j$, $i\not\in {N}_{j,2}$, $\angle (W_{ji}(e^{\hat{j}\omega})) =  \pi$ for all $\omega \in [0, 2\pi)$. Then
    \begin{align*}
    &\text{Real}(-{a_{ji}}{S_j^{*}(e^{\iota \omega})}\Phi_{p_j}^{-1}(e^{\iota \omega})
    -{a_{ij}}{S_i(e^{\iota \omega})}\Phi_{p_i}^{-1}(e^{\iota \omega}))>0,\nonumber\\
   &\text{Imag}(-{a_{ji}}{S_j^*(e^{\iota \omega})}\Phi_{p_j}^{-1}(e^{\iota \omega})
   - {a_{ij}}{S_i(e^{\iota \omega})}\Phi_{p_i}^{-1}(e^{\iota \omega}))=0\nonumber
    \end{align*}
    for all $\omega \in [0, 2\pi)$.
    \item Suppose $i$ and $j$ are such that $i \in {N}_j$ and $i \in {N}_{j,2}$ (both one and two hop neighbors), with $\angle (W_{ji}(e^{\hat{j}\omega})) = \pi$ for all $\omega \in [0,2\pi)$. Then
     \begin{align*}
        &\text{Imag}(-a_{ij}S_i(e^{\hat{j}\omega})\Phi_{p_i}^{-1}(e^{\hat{j}\omega})- a_{ji}S_j^{*}(e^{\hat{j}\omega})\Phi_{p_j}^{-1}(e^{\hat{j}\omega}))=0,\\
        &\text{and Real}(-a_{ij}S_i(e^{\hat{j}\omega})\Phi_{p_i}^{-1}(e^{\hat{j}\omega})-
        a_{ji}S_j^{*}(e^{\hat{j}\omega})\Phi_{p_j}^{-1}(e^{\hat{j}\omega}))\\
        &+\sum_{k \in N_j \cap N_i}a_{kj}a_{ki}\Phi_{p_k}^{-1}(e^{\iota \omega}) > 0, %\\ 
        \end{align*}
    for all $\omega \in [0,2\pi)$.
\end{enumerate}
\end{theorem}
\begin{proof}

Using (\ref{eqn:wienerphase}), (\ref{eqn:Phinv}) and theorem \ref{thm:notpiresult2}, the result can be verified. The proof is left to the reader.
\end{proof}
\begin{remark}
The consequence of Theorem \ref{thm:notpiresult2} is that for nodes $i$ and $j$ that are neighbors but not two hop neighbors, or, nodes $i$ and $j$ that are neighbors and two hop neighbors, $\angle(W_{ji}(e^{\hat{j}\omega})) = \pi$ for all $\omega \in [0,2\pi)$ is possible when the system parameters satisfy a restrictive and specific set of conditions. In other words, aside for pathological cases, the converse of Theorem \ref{thm:expression} holds. We use the phase response of the Wiener filters as a criteria to differentiate between true edges and spurious edges in the moral graph $\mathcal{G}_M$ (obtained by Wiener filtering) to recover $\mathcal{G}$.\end{remark}
\subsection{Learning Algorithm}
We now present Algorithm $1$ that estimates the topology of a RC network based on temperature deviation time series measurements from the nodes (zones). The algorithm consists of two parts. The first part (Steps \ref{step1_a} - \ref{step1_b}) determines the multivariate Wiener filter $W_{ji}(z)$ to estimate the moral graph $\mathcal{G}_M$. Based on Theorem \ref{thm:sparse_wiener}, the edge set $\bar{\mathcal{E}}_K$  is populated by adding a link between each node pair $i,j$ if the $H_{\infty}$ norm of $W_{ji}(z)$ or $W_{ij}(z)$ is greater than a predefined threshold $\rho$. Thus $\bar{\mathcal{E}}_K$ estimates one and two hop neighbors. Next (Steps \ref{step2_a} - \ref{step2_b}), we consider a finite set of frequency points $\Omega$ in the interval $[0,2\pi)$ and evaluate the phase angle of the Wiener filters corresponding to edges in $\bar{\mathcal{E}}_K$. Based on Theorem \ref{thm:expression}, if the phase angle is within a pre-defined threshold $\tau$ of $\pi$, the algorithm designates the edge between the concerned nodes as spurious edges and prunes them from $\bar{\mathcal{E}}_K$ to produce edge set $\bar{{\cal E}}$, which is an estimate of the edge set $\mathcal{E}$ of the true topology. In the limit of infinite data samples from each agent, $\bar{{\cal E}} = {\cal E}$.  From the implementation sake, we allow lags upto an order $F$ in (\ref{wiener}). In order to account for finiteness of measurements at each node, a regularized version of multivariate Wiener filtering in (\ref{wiener}) is used as shown below:
\begin{align}
\smashoperator[r]{\inf_{h_{ji}\forall i\neq j}}~\mathbb{E}(T_j(k)-
		\smashoperator[lr]{\sum_{i=1,i\neq j}^{m}}~\sum_{L=-F}^{F}h_{ji}^{L}T_i(k-L))^2 +\gamma\smashoperator[lr]{\sum_{i=1,i\neq j}^{m}}
	\|h_{ji}\|_1, \label{regularized}
\end{align}

Here $\gamma$ is the regularization parameter and $h_{ji}=[h_{ji}^{-F},...,h_{ji}^{0},...,h_{ji}^{F}]$ and the Wiener filter is 	\begin{align}\label{regularfilter}
	    W_{ji}(z) = \sum_{L=-F}^{F}h_{ji}^{L}z^{-L}. 
\end{align}
\begin{algorithm}
\caption{RC network topology learning using multivariate Wiener Filtering}
\textbf{Input:} Time series of temperature deviation  $T_i(k)$ for each zone $i \in \{1, 2,... m\}$ and time-step $k$ in a building. Thresholds $\rho,\tau$. Frequency points $\Omega$. Regularization parameter $\gamma$.\\
\textbf{Output:} Reconstruct the true topology with an edge set $\bar{{\cal E}}$ \\
\begin{algorithmic}[1]
\ForAll{$l \in \{1,2,...,m\}$}\label{step1_a}
\State Compute the Wiener filter ${W}_{lp}(z)$ using the temperature time-series $\forall p \in \{1, 2,... m\} \setminus {l}$ \label{step1_a1}
\EndFor
\State Edge set $\bar{\cal E}_K \gets \{\}$ 
\ForAll{$l,p \in \{1,2,...,m\}, l\neq p$}
\If{$H_{\infty}({W}_{pl}(z)) > \rho$}\label{step_nonzero}
\State $\bar{\mathcal{E}}_K \gets \bar{\mathcal{E}}_K \cup \{(l,p)\}$
\EndIf
\EndFor\label{step1_b}
\State Edge set $\bar{\cal E} \gets \bar{\cal E}_K$ \label{step2_a}
\ForAll{$l,p \in \{1,2,...,m\}, l\neq p$}
\If{$\pi -\tau \leq |\angle(W_{pl}(e^{\hat{j}\omega}))| \leq \pi, \forall \omega \in \Omega$}
\State $\bar{\mathcal{E}} \gets \bar{\mathcal{E}} - \{(l,p)\}$
\EndIf
\EndFor \label{step2_b}
\State Error = $\frac{\text{Number of false edges}}{\text{Number of true edges}}$ 
\end{algorithmic}
\end{algorithm}

\section{Results}
In this section, we discuss performance of our learning algorithm in estimating the RC network topology of a 5 zone building using zonal temperatures obtained through EnergyPlus simulation. The building consists of a core zone and four perimeter zones as shown in Fig. \ref{fig:energyplus_5node}(a). Its RC network is obtained from {\cite{6315699}} and shown in Fig. \ref{fig:energyplus_5node}(b). From the true topology of the RC network (See \ref{fig:energyplus_5node}(c)), it is clear that every neighbor of a node is also its two hop neighbor. Thus the number of spurious links detected in the moral graph (See \ref{fig:energyplus_5node}(d)) are substantial and phase based criteria of Algorithm $1$ become crucial for exact reconstruction. The floor details of the zones in the building are Core: 149.66 sq. m., Perimeter I: 113.45 sq.m., Perimeter II: 67.3 sq.m., Perimeter III: 113.45 sq.m., Perimeter IV: 67.3 sq.m. and the height of the building is 3.05 m. The building location is Minneapolis, MN and the weather file used in EnergyPlus is obtained from \url{https://energyplus.net/weather}. The exogenous inputs to the system include lighting and electrical load, and schedules of people in the building. For our simulations, we consider two models of the electrical loads and lighting: (a) white Gaussian, and (b) time-correlated wide sense stationary processes. The correlated input values are  obtained  by  filtering  white  Gaussian  noise  through  1D digital  filter  in  MATLAB. The effect of solar radiation is common to all the zones and can be filtered. The EnergyPlus output comprises of zone air temperatures which is the instantaneous temperature averaged over space for each zone. This temperature data is obtained with one minute granularity and used for topology reconstruction for both white and correlated input distributions.

As detailed in Algorithm $1$, we first obtain the neighbor and two hop neighbor set by inspecting the $H_{\infty}$ norm of the Wiener filters. Figs. \ref{fig:mag}(a) $\&$ (b) shows the $H_{\infty}$ norms of the filters $W_{2i}$'s and $W_{3i}$'s between nodes $2, 3$ and all other nodes $i$ for white Gaussian input. Observe that the magnitude of $H_{\infty}$ norm of $W_{24}$ is relatively large despite $2,4$ being two hops away for both white and colored inputs. Hence it is clear that magnitude itself cannot be used to distinguish between true and two hop neighbors. Next we use the phase response of the filters to remove spurious links. In Figs. \ref{fig:phase}(a) and (b), we can clearly observe that the phase of $W_{24}$ is close to $\pi$ for both white and colored inputs. This indicates that nodes $2$ and $4$ are two hop neighbors and the link between them is pruned.
\begin{figure}[tb]
	\centering
	\begin{tabular}{cc}
		\includegraphics[width=0.45\columnwidth, height = 3 cm]{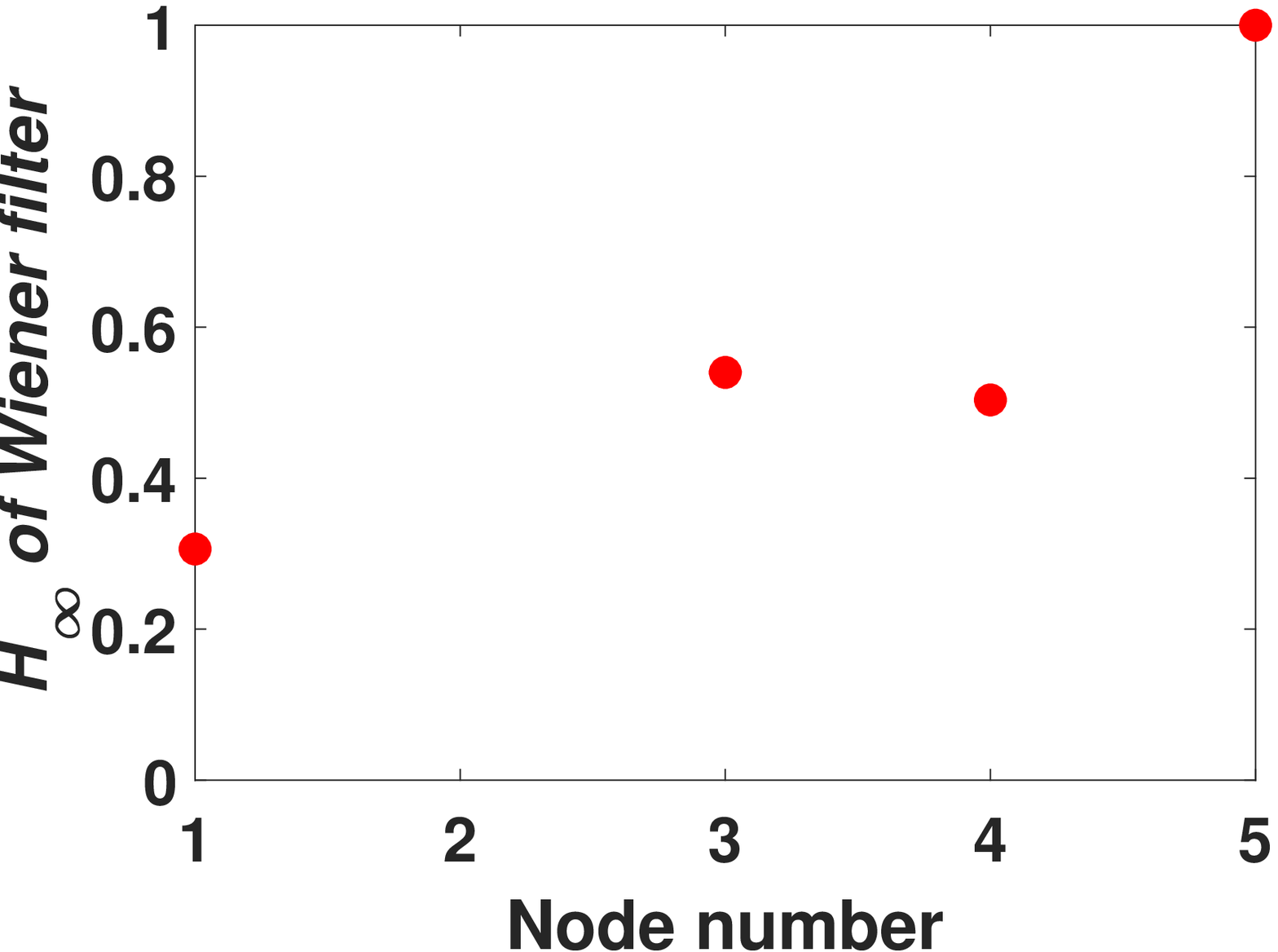} &
		\includegraphics[width=0.45\columnwidth, height = 3 cm]{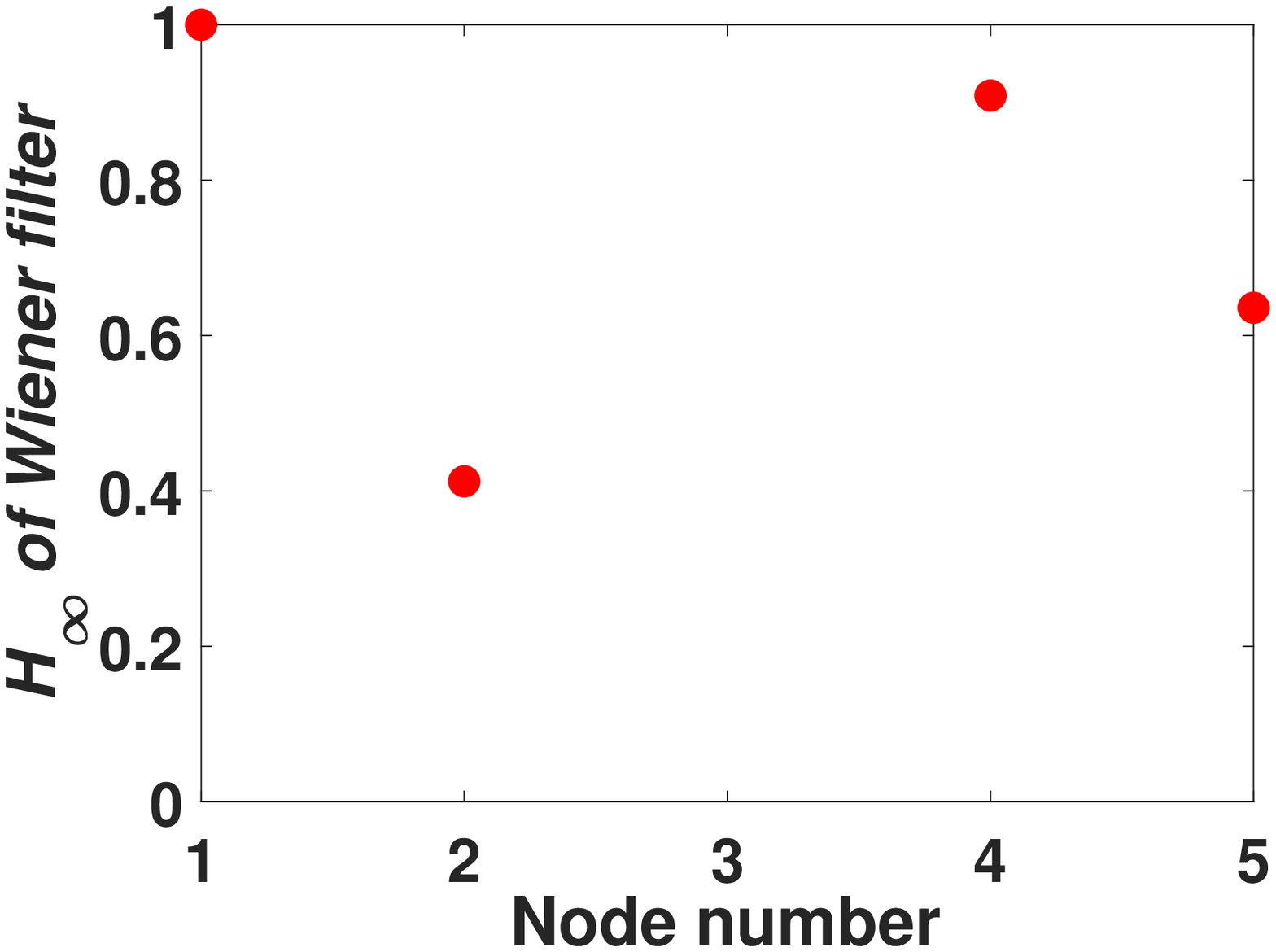}\\% &
		(a) & (b)\\
	\end{tabular}
	\caption{$H_{\infty}$ norm for Wiener filters between nodes $2,3$ and all nodes for different inputs distributions (sample size = $10^5$) (a) Node $2$, white Gaussian inputs, (b) Node $3$, white Gaussian inputs. \label{fig:mag}}
\end{figure}
\begin{figure}[tb]
	\centering
	\begin{tabular}{cc}
		\includegraphics[width=0.45\columnwidth, height = 3 cm]{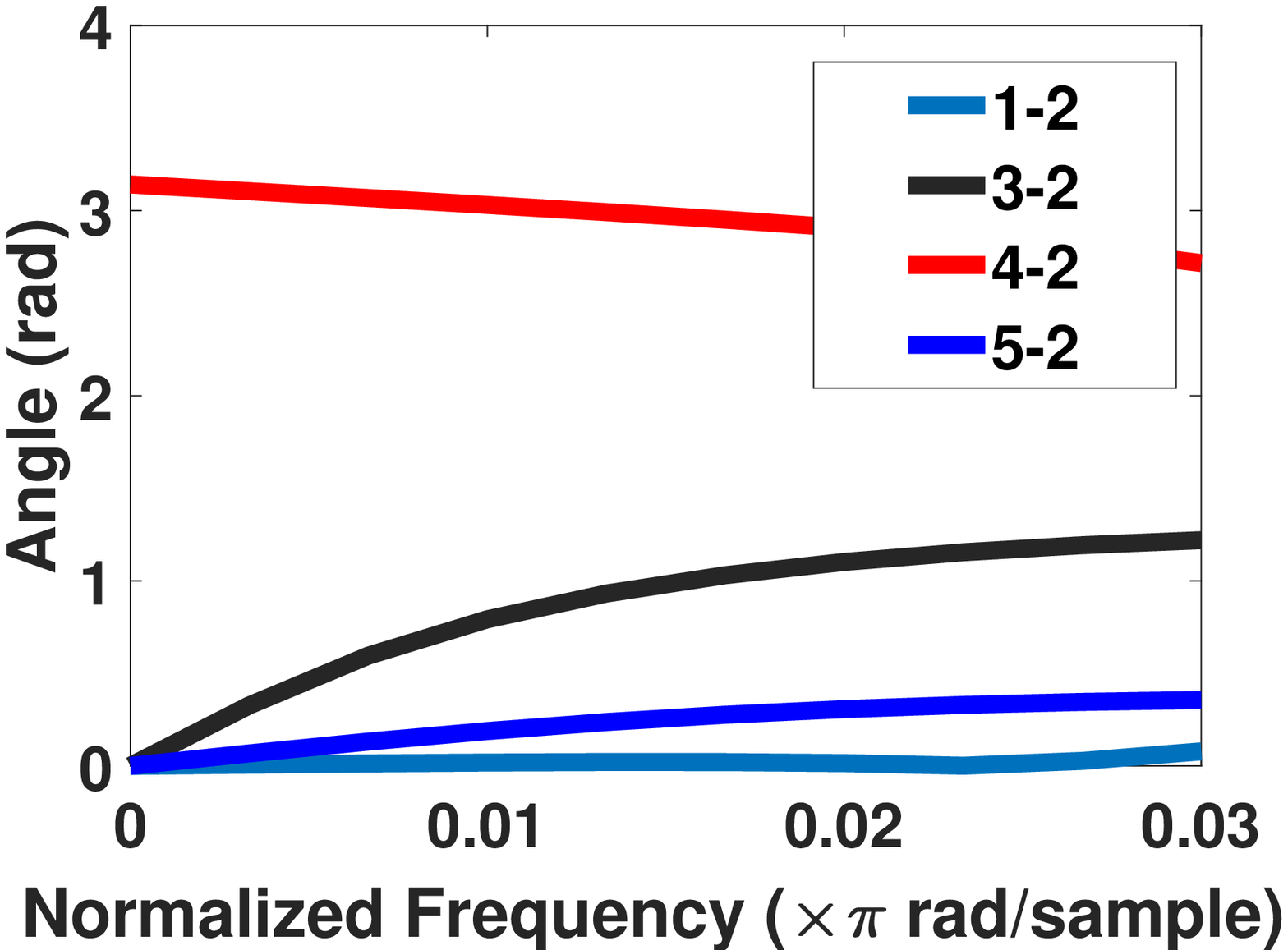} &
		\includegraphics[width=0.45\columnwidth, height = 3 cm]{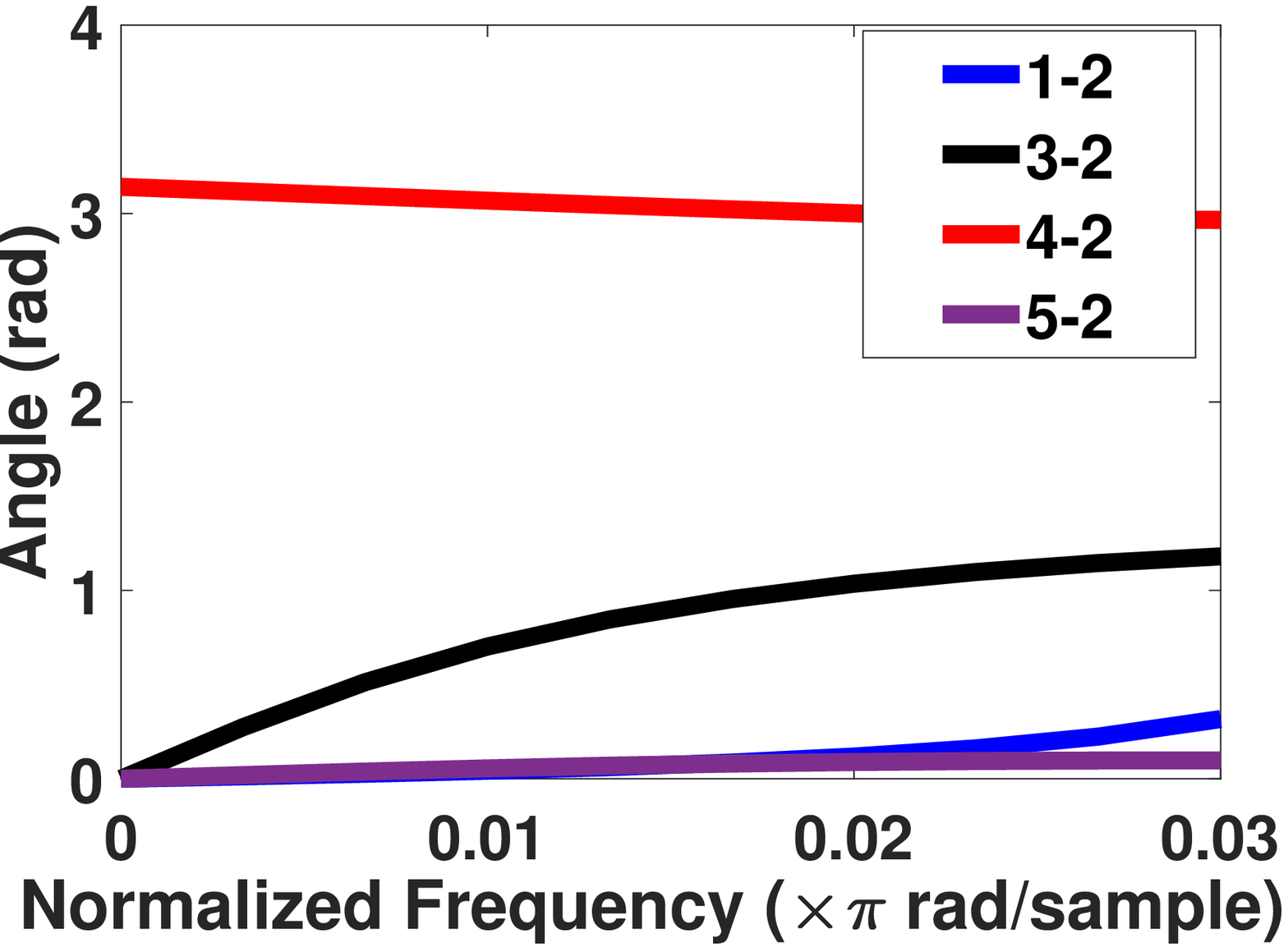}\\% &
		(a) & (b)
	\end{tabular}
	\caption{Absolute phase values for Wiener filters between node $2$ and all nodes for different inputs distributions (sample size = $10^5$) (a) white Gaussian inputs, (b) WSS inputs\label{fig:phase}}
\end{figure}
We present the average errors produced by our learning algorithm for different temperature samples sizes in Fig \ref{fig:errorplot}(a) and (b) for white Gaussian and colored inputs respectively. Further we demonstrate errors obtained by using the regularized version of the Wiener filter reconstruction (see (\ref{regularized})). Note that while the errors in either case goes down to zero (exact recovery), the performance of the regularized approach is better at low sample sizes owing to the use of regularizers. The results are particularly noteworthy given that EnergyPlus simulations need not satisfy all the theoretical assumptions necessary for consistency outlined in Theorem~\ref{thm:notpiresult2}. We compare our learning method with that of an one-step ahead regression framework \cite{pereira2010learning} used for estimation in LTI systems with white Gaussian inputs. The problem formulation for regression is as follows:

\begin{align}
\smashoperator[lr]{\inf_{{\{h_{ji}\}},{\forall i}}} \mathbb{E}(T_j(k)-\sum_{i=1}^{m}h_{ji}T_i(k-1))^2 +\gamma\sum_{i=1}^{m}	\|h_{ji}\|_2, \label{regularized_regression}
\end{align}
where $\gamma$ is the regularization parameter. Here, if the magnitude of $h_{ij}$ that solves (\ref{regularized_regression}) is greater than a specified threshold then a link between $i \& j$ is assessed as being present. In contrast to our learning algorithm, the topology reconstruction using regression performs poorly for both white Gaussian and colored inputs as shown in Fig. \ref{fig:errorplot}(a) and (b). This is due to the fact that the estimated coefficients in the regression problem are non-zero for non-neighbors. Figs. \ref{fig:coeff_regression}(a) and (b) show the estimated coefficients ($h_{2i}$ and $h_{3i}$) between nodes $2,3$ respectively and all other nodes for colored inputs. Note that $2,4$ are non-neighbors but regression based estimation yields a non-zero value. It is evident from Fig. \ref{fig:errorplot} that no matter how large the data set is, the reconstructed topology based on regression deviates from true topology by at least 50$\%$.

\begin{figure}[tb]
	\centering
	\begin{tabular}{cc}
		\includegraphics[width=0.46\columnwidth, height = 3cm]{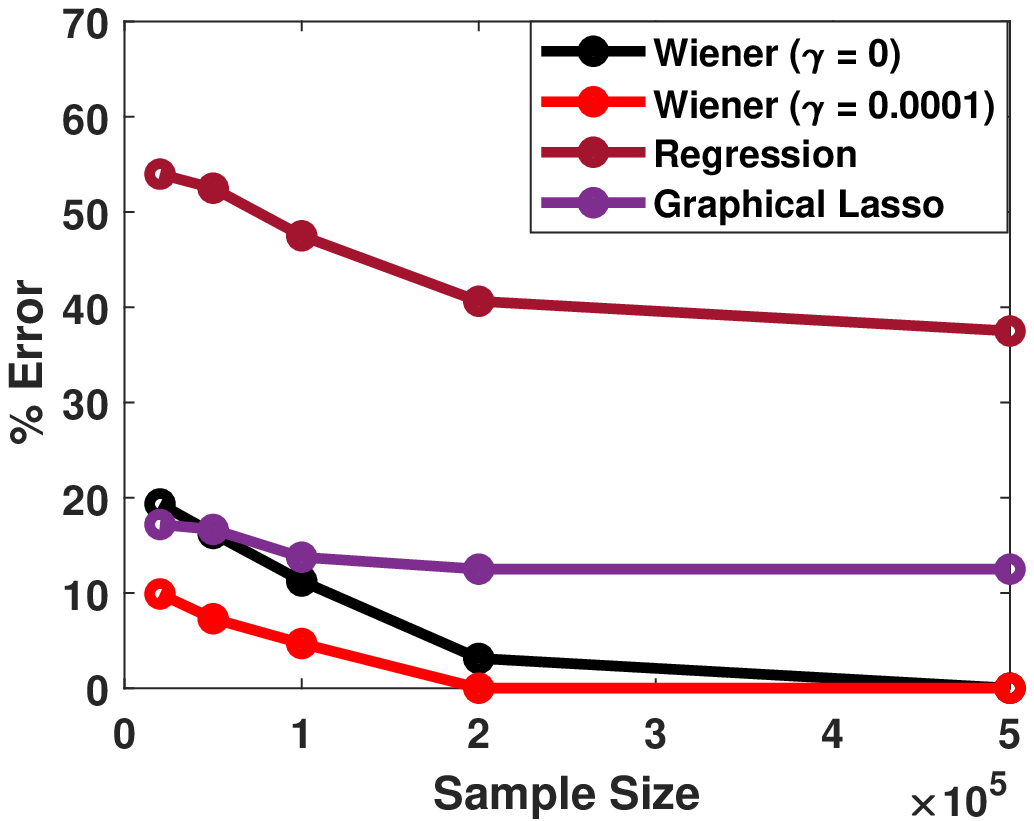} &\includegraphics[width=0.46\columnwidth, height = 3 cm]{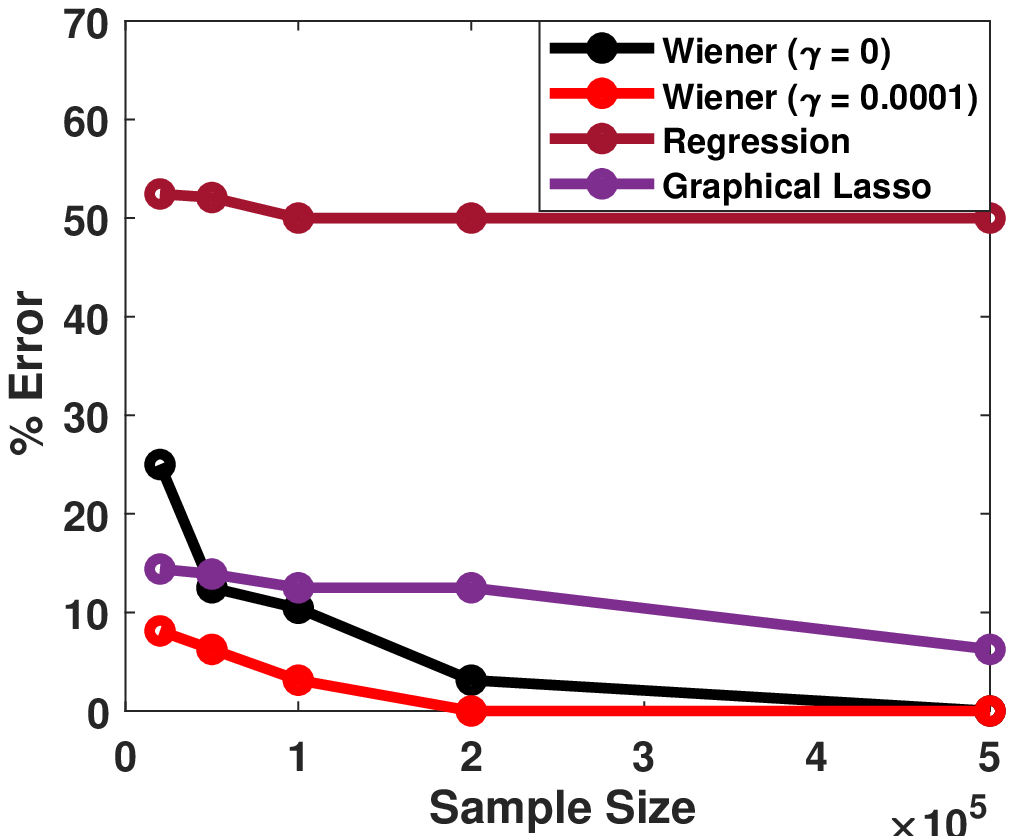}\\
		(a) & (b)
	\end{tabular}
\caption{Error percentage variation with number of samples per node for Algorithm $1$ (with and without regularizers), Regression \cite{pereira2010learning}, and GLASSO \cite{goyal2010modeling} when inputs are (a) White Gaussian (b) WSS input \label{fig:errorplot}}
\end{figure}
\begin{figure}[tb]
	\centering
	\begin{tabular}{cc}
		\includegraphics[width=0.45\columnwidth, height = 3 cm]{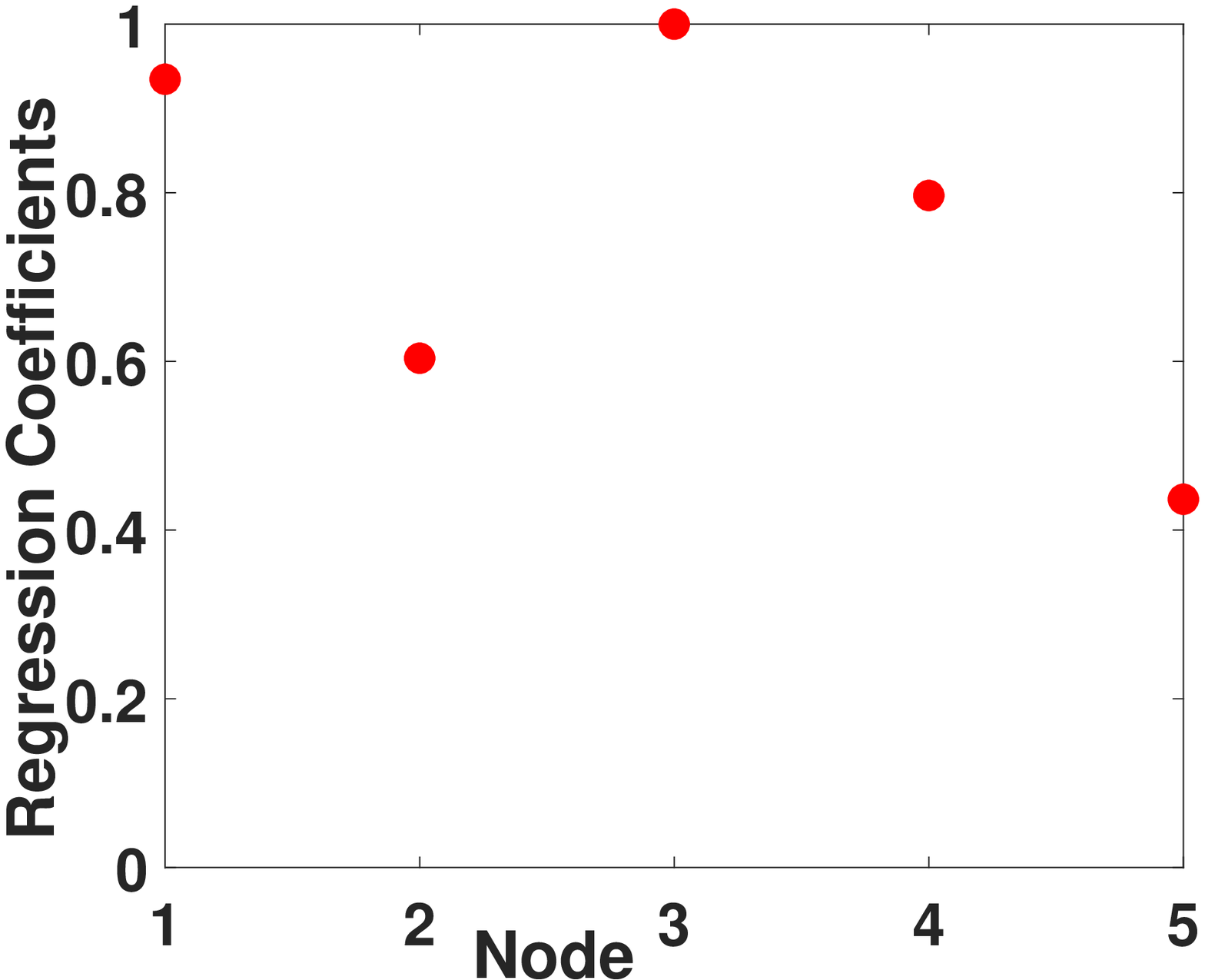} &
		\includegraphics[width=0.45\columnwidth, height = 3 cm]{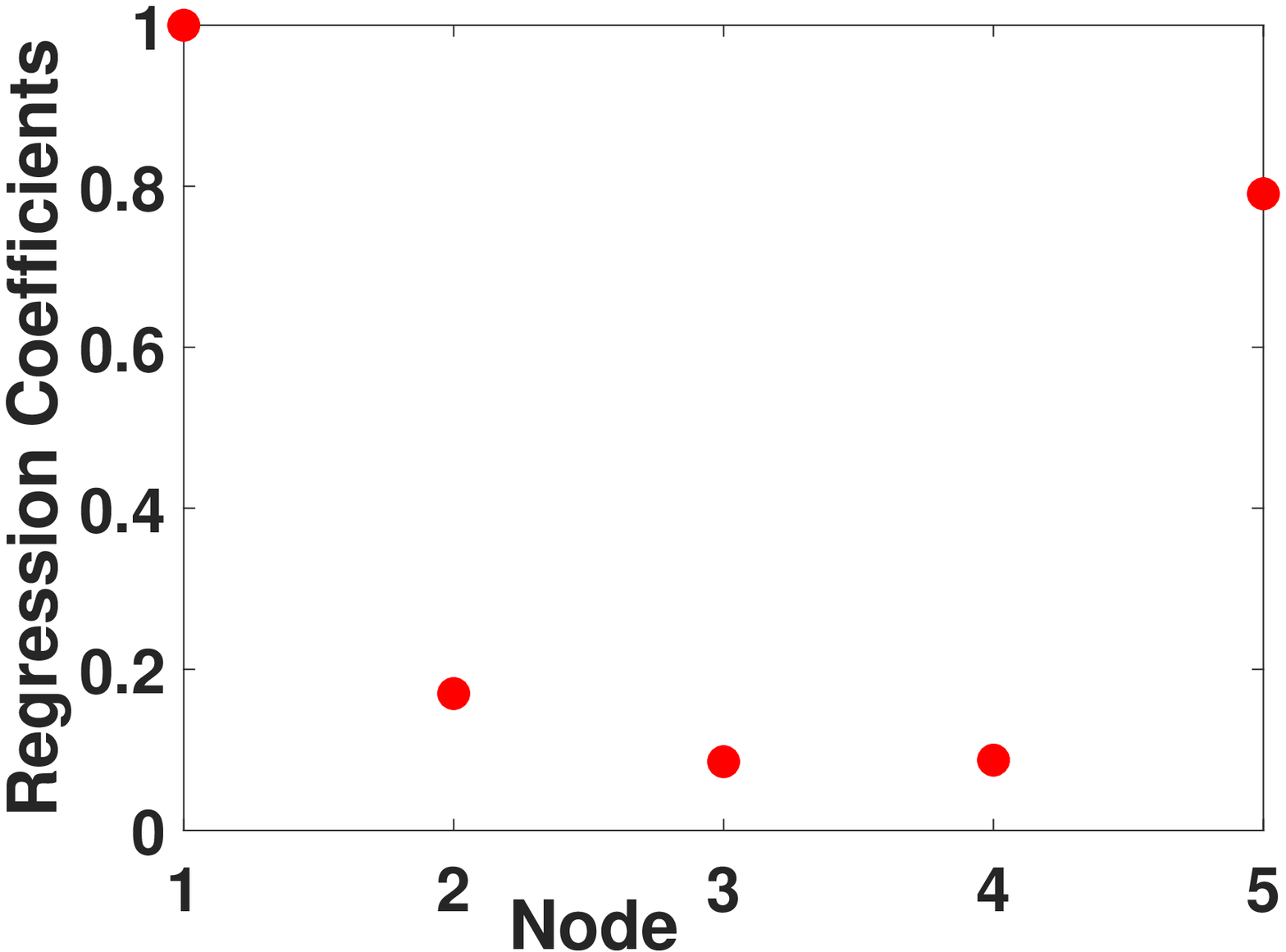}\\% &
		(a) & (b)
	\end{tabular}
% 	\caption{An oriented graph (a) and its kin topology 	\end{tabular}
	\caption{Coefficients recovered by regression (\ref{regularized_regression}) problem with WSS inputs for (a) Node 2 (b)Node 3 \label{fig:coeff_regression}}
\end{figure}
Finally, we also compare our learning algorithm with Graphical Lasso (\cite{friedman2008sparse}, \cite{drton2004model}) meant for static Gaussian graphical models that has been utilized for structure estimation \cite{6161387}. From Figs. \ref{fig:errorplot}(a) and (b), we observe that the performance of our learning algorithm is better than that of Graphical Lasso under both input regimes. 

\section{Conclusion}
In summary, we have provided analytic results with provable guarantees on near exact reconstruction of topology of the thermal dynamics of buildings. The resulting algorithm is an effective data driven approach for system identification of thermal dynamics of a building using RC network. The learning algorithm recovers the exact interaction topology of a group of zones and buildings. The error in reconstruction decreases as more and more samples are used for learning the interaction topology. In the limit of large number of samples, the presented algorithm recovers the exact interaction topology of the building, which has potential applications in HVAC control, demand response and building safety. The data is generated using EnergyPlus and the topology reconstruction is implemented in Python. Note that in situations of temperature data breach, the building structure details can be figured solely from the data and hence could pose a threat. Parameter estimation of the identified RC network after accounting for wall capacitance and applications in data-driven control and cyber security of buildings will be the focus of future work.

\section{Acknowledgments}
The authors H. Doddi, S. Talukdar, and M. V. Salapaka acknowledge the support of ARPA-E for supporting this research through the project titled \lq A Robust Distributed Framework for Flexible Power Grids\rq \ via grant no. DE-AR$0000701$ and Xcel Energy's Renewable Development Fund. D. Deka acknowledges the support of funding
from the Center for Non-Linear Studies at LANL and the Grid Modernization Initiative of the U.S. Department of Energy's Office of Electricity.

%\bibliography{topident,l0problem,control}

%\addtolength{\textheight}{-12cm} 
%\vspace{-5 cm}

\end{document}